\RequirePackage{fix-cm}

\documentclass[a4paper, 11pt]{article}     

\usepackage{graphicx}
\usepackage{latexsym}
\usepackage{amsmath,amssymb,amsfonts,makeidx,dsfont}

\usepackage[margin = 3cm]{geometry}

\usepackage{natbib}
\usepackage{color}
\usepackage{lscape} 
\usepackage{lineno} 
\usepackage{color}
\DeclareGraphicsExtensions{.bmp,.png,.pdf,.jpg,.eps}

\newtheorem{theo}{Theorem}

\newtheorem{proof}{Proof}

\newcommand{\Prob}{\mathbb{P}}
\newcommand{\Expec}{\mathbb{E}}
\newcommand{\Var}{\text{Var}}
\newcommand{\Cov}{\text{Cov}}
\newcommand{\histo}{\hat{f}^{\textsc{h}}_n}
\newcommand{\histom}{\bar{f}^{\textsc{h}}_n}
\newcommand{\kde}{\hat{f}^{\textsc{kde}}_n}
\newcommand{\fp}{\hat{f}^{\textsc{fp}}_n}
\newcommand{\fstar}{\hat{f}^*}
\newcommand{\bagfp}{\hat{f}^{\textsc{bagfp}}_n}
\newcommand{\baghist}{\hat{f}^{\textsc{baghist}}_n}
\newcommand{\bagkde}{\hat{f}^{\textsc{bagkde}}_n}

\begin{document}

\title{Bagging of Density Estimators}

\author{Mathias Bourel \and Jairo Cugliari}

\maketitle

\begin{abstract}
In this work we give new density estimators by averaging classical density estimators such as the histogram, the frequency polygon and the kernel density estimators obtained over different bootstrap samples of the original data. We prove the $L^2$-consistency of these new estimators and compare them to several similar approaches by simulations. Based on them, we give also a way to construct non parametric pointwise confidence intervals for the target density.\\
\textbf{keywords}: {density estimation \and aggregation \and bagging \and histogram \and polygon frequency \and kernel density estimator. }

\end{abstract}

\section{Introduction}\label{sec:intro}

Ubiquitous in data analysis, density estimation techniques are certainly the most
used unsupervised learning technique on low dimension. 
Whether for studying asymmetry, normality, residual diagnostic or bump hunting among others, one usually relies 
on a visual inspection of a plot of the density to take a primary decision, mostly in one or two dimensions.

The general aim is to gather
basic information about the unobserved data generation mechanism out of a sample
of $n$ observations say $x_1, \dots, x_n$. One usually supposes that the observations are realizations of a random variable $X$
that admits a probability density function $f$ (i.e. $f$ is non negative and integrates
1). Then, the learning task is to estimate $f$ as accurately as possible. First, by
obtaining a point wise estimate $\hat f(x)$ of $f(x)$ for all $x \in \mathbb{R}^d$, and second, by assessing the uncertainty of the point wise estimate through the
construction of a confidence interval for $f$. Of course both problems implies
different difficulties and involves specific techniques. In what follows, we focus on 
nonparametric approaches for the first objective and give a possible way to construct a 
pointwise confidence interval for the true density.
In particular we center our attention on three classes of \textit{base} or individual estimators for the density: 
histograms, frequency polygons and the kernel density estimator. We postpone the
formal definition up to the next section but we provide a discussion guided by 
intuitive descriptions here. 

Histograms are undoubtedly the most popular construction for density estimation.
They rely on the best constant-wise approximation of $f$ given by the data
using a binning argument. 
The power of this simple construction combined with the ease of its interpretation
makes them accessible to non technical users. Besides, theoretical properties
can be derived showing that histograms are consistent estimators.
Some of the lacks of the histograms are inherent to the constant level for at each partition, 
on one hand their discontinuities and on the other hand they have null derivative everywhere. 
Frequency polygons are constructed on top of histograms (and so take profit of the
binning advantages) providing a linear piecewise estimator.
Although the added regularity was once pointed out as a flaw (see \cite{fisher1932})
it is now well known that it increases the quality of the estimator. Theoretically
this is shown by a faster rate of convergence (cf. section \ref{sec:art}).
Further regularity can be gained using kernel density estimators. Essentially
one first picks a kernel function with the desired degree of regularity for the
final estimate. Then, the empirical measure is convolved to produce the kernel 
density estimation of $f$. 
It has been proved in \cite{scott1985} that the rates of convergence of the frequency polygon are 
similar to the kernel density estimator.
All the three approaches are exhaustively studied both from on practice and theory
as individual estimators and a general reference for the subject is \cite{scott2015multivariate}. 
However, a reasonable question is to ask whether further
improvement can be achieved from these construction by means of aggregation schemes.

Ensemble learning or aggregation methods are increasingly used in the supervised framework: these methods combine intermediate predictors to obtain an aggregated model with the aim to obtain a better estimator. Bagging \citep{Brei-bag}
(Bootstrap and AGGregatING), Boosting \citep{FSc}, Stacking \citep{Wolpert}, and Random Forests \citep{Brei-RF} have been broadly studied in the case of classification (principally binary classification) or regression from the theoretical viewpoint and have very high performances when tested over tens of various datasets selected from the machine learning benchmark. Several extensions are still under study: multivariate regression, multiclass classification, and adaptation to functional data or time series. Very few developments exist for ensemble learning for unsupervised techniques such as clustering analysis and density estimation. Only on few works several authors look at the adaptation of the aggregation procedure to estimate a density under somehow restrictive conditions. One of the first is the mean of the most simplest density estimator, the histograms, each one constructed over several different deterministic grids in Average Shifted Histograms, ASH \citep{scottash}. With a combination of several kernel density estimators with different bandwidths, often in a normal context, and varying the form of the aggregation we can cite  \cite{Ridgeway}, \cite{Glodek2013}, \cite{Song}, \cite{Rosset}, \cite{Smyth} and \cite{Rigo}.

Another kind of aggregation can be obtained by introducing randomness in the individual estimators. In \cite{bourel-ghattas} 
the authors include randomness in the construction of the intermediate histograms using then different aggregation
schemes: AggregHist, using simple aggregation, BagHist using Bagging or StackHist using Stacking. 
Mathematically well sound, these approaches were explored thorough empirical simulation without theoretical framework. 
The Random Average Shifted Histogram, RASH \citep{Bourel13} is constructed as the mean average of different histograms, each one constructed over a random translated grid of the initial breakpoints of the initial histogram. RASH is show to be
consistent and perform well in comparison to the precedent aggregation schemes.

In this work, we explore the contribution of Bagging to the density estimation task (like for BagHist) where the intermediate estimators are either histograms, frequency polygons or kernel density estimators. 
This article is organized as follow. Section \ref{sec:art} introduces notation and reviews the framework context we need about density estimators. In Section \ref{sec:bagfp} we present our three methods: BagHist, BagFP and BagKDE and establish a theoretical result about their consistency. The results of comprehensive simulations are the object of Section \ref{sec:experiments}. For this, we use several target densities available on literature to evaluate the performance of our estimators and comparing them to classical density estimators. We also explore the construction of pointwise confidence intervals as an approximation of a confidence band using the bootstrap procedure. The work concludes with a discussion in Section \ref{sec:discussion}.

\section{Some density estimators\label{sec:art}}

We look at the principal results using different density estimators. While for the histogram and 
the kernel density estimator these results are quite popular we give some detail in order to fix notation.
For a detailed exposition see \cite{scott2015multivariate}.

In all cases, the starting point is an independent and identically distributed sample $x_1,\dots, x_n$ of a real random variable with density $f$.

\begin{itemize}
\item \emph{Histogram}.
Let $B_{j}=\left[jh,(j+1)h\right]$ be a set of intervals defined over the support of $f$, and $h \to 0$ as $n \to +\infty$.

The ordinary histogram is defined as:
\begin{equation} \histo(x)= \frac{1}{nh} \sum \limits_{i=1}^{n} \sum \limits_{j \in \mathbb{Z}}
\mathds{1}_{B_{j}}(x_i)\mathds{1}_{B_{j}}(x) =\frac{\nu_j}{nh} \label{eq:histo}
\end{equation}
where $\nu_j \sim \text{Bin}(n,p_j)$, with $p_j=\int_{B_{j}} f(t)dt$, is the number of observations of the sample that fall in bin $B_{j}$.
If $x \in B_{j}$, we have that:
$$\mathds{E}\bigr(\histo(x)\bigl)=\frac{1}{nh}\sum \limits_{i=1}^{n}
\sum \limits_{j \in\mathbb{Z}} \mathds{P}(x_i \in B_{j})\mathds{1}_{B_{j}}(x) =\frac{1}{nh} n \mathds{P}(x_i \in
B_{j})=\frac{p_j}{h}=f(\xi_j)$$

$$\Var\bigl(\histo(x)\bigr)= \frac{1}{n^2h^2}\Var
(\nu_j) =
\frac{p_j(1-p_j)}{nh^2}\leq \frac{p_j}{nh^2}=\frac{f(\xi_j)}{nh}$$

for some $\xi_j \in B_j$. For a $x \in B_{j}$ a fixed point, when $h \to 0, n \to \infty$ and $nh \to \infty$, we get the classical properties for the histogram:
$$\mathds{E}\bigr(\histo(x)\bigl) \to f(x), \ \Var\bigl(\histo(x)\bigr) \to 0.$$
and, moreover, if $f$ is locally Lipschitz the histogram  is mean square consistent, i.e $\text{MSE}\left(\histo(x)\right)=\text{Bias}^{2}\left(\histo(x)\right)+\Var\left(\histo(x)\right)\to 0$.

It is widely used in many fields, because of its computational simplicity. The histogram depends on two  parameters: the bin width $h$ and an origin $x_0$ to fix the grid.
There is a huge literature that proposes several optimal choices for
$h$ using different criteria. If we suppose that the underlying density $f$ is Gaussian, it can be shown (see \cite{scott1979}) that an optimal choice for $h$ is of order $n^{-1/3}$. With this value the histogram has a rate of convergence of order $n^{-2/3}$ with respect to the Mean Integrated Squared Error (MISE).

\item \emph{Frequency Polygon}. Frequency polygons are constructed on top of histograms connecting with straight lines the midpoint of two consecutive bin values. The expression of the frequency polygon for an $x \in B_j=\left[ (j-1/2)h, (j+1/2)h\right]$ is 
\begin{equation} \label{eq:fp}
  \fp(x) = \left(\frac{1}{2}+j-\frac{x}{h} \right)\frac{\nu_j}{nh} 
                 + \left(\frac{1}{2}-j + \frac{x}{h} \right)\frac{\nu_{j+1}}{nh}
\end{equation}
The frequency polygon was deeply studied in \cite{scott1985}. With respect to the histogram, it
presents the the advantages of being continuous and smooth. Under weak conditions, an optimal choice for $h$ is
of order $n^{-1/5}$ and it achieves a rate of convergence of order order $n^{-4/5}$ with respect to the MISE (\cite{scott1985}).

\item \emph{Kernel Density Estimators}. A \emph{Kernel Density Estimator, KDE,} is a function defined by \begin{equation}\kde(x)=\frac{1}{nh} \sum \limits_{i=1}^n K\left(\frac{x-x_i}{h} \right) \label{eq:kde} \end{equation} for all $x\in \mathbb{R}$ where $K$ is a kernel function, i.e a non-negative, symmetric and unimodal function such that $\int K(u)\,du =1$. Parameter $h$ is called the bandwidth of the estimator $\kde$, who inherits all the mathematical properties of $K$. The function $K$ indicates the weight that observation $x_i$ has in the estimation of $x$: observations close to $x$ are weighted more important. It can be shown that $$\Expec\left(\kde(x)\right)=f(x)+\frac{f''(x)}{2}\mu_2(K)h^2+o(h^2)$$
$$\Var\left(\kde(x)\right)=\frac{||K||_2^2f(x)}{nh} + o\left(\frac{1}{nh}\right)$$
where $\mu_r(K)=\int u^rK(u)\,du$. 
When $h \to 0$ and $nh \to \infty$, we get the classical properties as for the histogram:
$$\mathds{E}\bigr(\kde(x)\bigl) \to f(x), \ \Var\bigl(\kde(x)\bigr) \to 0.$$
As for the histogram, if $h$ is large the variance decreases but the bias is large. On the other hand, if $h$ is small, the bias is small but the variance is large. The optimal rate of convergence of KDE is of order 
$n^{-4/5}$ as for the frequency polygon.
\end{itemize}

\section{Bagging of estimators}\label{sec:bagfp}

The bootstrap method was introduced in \cite{efron1979} and have the purpose of doing statistical inference using resamples of the original set of data. More precisely, if we have a data set $\mathcal{L}=\{x_1,x_2,\dots,x_n\}$ with distribution $F$, the non parametric bootstrap procedure consists to draw, with replacement, a new sample $\mathcal{L}^*=\{x_{1}^{*},x_{2}^{*},\dots, x_{n}^{*}\}$ from $\mathcal{L}$ of the same size. Then, the sample $\mathcal{L}^*$ has a
distribution $F_n$, the empirical distribution of $\mathcal{L}$. 
The main idea is that the sample $\mathcal{L}^*$ is to the original sample $\mathcal{L}$ what the sample $\mathcal{L}$ is to the population, so the method treats the empirical of a distribution of sample data as the true distribution. 
 
Reiterating this procedure several times and obtaining many bootstrap samples is a cornerstone
to the construction of several bootstrap based approaches. 

The bootstrap has three big applications: bias correction, construction of confidence interval and hypothesis testing \citep{efron}. 

However, the use of bootstrap in nonparametric density estimation requires some caution, particularly concerning the bias of the estimators. In our setting, let $\widehat{f}$ be a nonparametric density estimator for $f$ obtained from the sample $\mathcal{L}$.

Now we draw a bootstrap sample $\mathcal{L^{*}}$ of $\mathcal{L}$. 

In the case of the kernel density estimator, the estimation over $\mathcal{L}^{*}$ is without bias,
\begin{align*}
\mathds{E} \left(  \hat{f}^{*}(x)|\mathcal{L}\right)=& \frac{1}{nh} \sum \limits_{i=1}^{n} \mathds{E}\left(K\left( \frac{x-x_{i}^{*}}{h}\right)\right)\\
=& \frac{1}{h} \mathds{E}\left(K\left( \frac{x-x_{i}^{*}}{h}\right)\right)=\frac{1}{h}  \sum \limits_{y \in R_{x_{i}^{*}}} K\left( \frac{x-y}{h}\right)\mathds{P}(x_{i}^{*}=y)\\
=&\frac{1}{h} \sum \limits_{j=1}^{n} K\left( \frac{x-x_j}{h}\right)\frac{1}{n}=\frac{1}{nh} \sum \limits_{j=1}^{n} K\left( \frac{x-x_j}{h}\right)=\hat{f}(x).
\end{align*}
This simple results holds also for the histogram and for the frequency polygon too (see demonstration of theorem \ref{teo}), that is $\mathds{E}(\widehat{f}^{*}(x)|\mathcal{L})=\widehat{f}$.
The fact is treated in \cite{hall1997bootstrap} in this terms: ``far from accurately estimating the substantial bias of $f$, the bootstrap sets the bias of this kind of density estimator equal to zero''. 

In supervised learning, the main application of bootstrap is definitely the Bagging.
It is a parallel aggregation method of individual entities that at each step $b$ draw a bootstrap sample $\mathcal{L}^{*}_{b}$ of the original sample $\mathcal{L}$ and compute an estimator (a classifier in classification or a predictor in regression) over $\mathcal{L}^{*}_{b}$. For an input $x$, the output of the Bagging method is the average in regression or the majority rule in classification of the intermediate estimators at $x$. We follow this aggregation strategy to construct new density estimators of a density function $f$. Our procedures run as follows (Figure \ref{algoritmo}):

\begin{figure}[ht]
\begin{center}
\fbox{
\begin{minipage}{0.9\textwidth}
 Let $\mathcal{L} = \{x_1, \dots,x_n\}$ be a sample with unknown distribution $F$ admitting a 
density $f$. Also, considerer $\widehat{f}_n$ a density estimator evaluated in $\mathcal{L}$. 

For $b \in 1, \dots, B$:

\begin{enumerate}
\item obtain $\mathcal{L}^{*}_{b}=\{x_1^{*}, \dots, x_n^{*}\}$ a bootstrap sample from $\mathcal{L}$; 
\item construct $\fstar_b$ the density estimator obtained over this bootstrap sample. In particular the bandwidth $h$ is calculated over $\mathcal{L}^{*}_{b}$ 
\end{enumerate}

{\bf Output:} The final estimator is the simple pointwise average of the individual estimators 
i.e. $$\widehat{f}^{*} (x)= \frac{1}{B} \sum \limits_{b=1}^{B} \fstar_b (x)$$

\end{minipage}}
\end{center}
 \caption{Bagging of density estimators} \label{algoritmo}
 \end{figure}

To obtain the bagged histogram (BagHist) estimator $\baghist$ we simply use at each step $b$, histograms $\histo$ defined in \eqref{eq:histo} as $f_{b}^{*}$. Analogously for bagged frequency polygons (BagFP) $\bagfp$ and bagged kernel density estimators (BagKDE) $\bagkde$, we replace  $f_{b}^{*}$ with frequency polygon estimator $\fp$ (cf. Eq. \eqref{eq:fp}) or kernel density estimator $\kde$ (cf. Eq. \eqref{eq:kde}) respectively.
 
\subsection{$L^2$ consistency of the Bagging of estimators }
Here we prove $L^2$ consistency, for three estimators BagHist, BagFP and BagKDE which correspond to bagging of histograms, bagging of polygon frequencies and bagging of kernel density estimators proving that, since $h \to 0, n \to + \infty$ and $nh \to +\infty$, for all $x$ in bin $B_j$ (in case of histogram or frequency polygon) or for all $x \in \mathbb{R}$ (in case of kernel density estimator): $$
\Expec [(\baghist (x) - f(x))^2] \to 0,\,\,\,\,\,
\Expec [(\bagfp (x)   - f(x))^2] \to 0, \,\,\,\,\,$$
$$\Expec [(\bagkde (x)  - f(x))^2] \to 0.
$$

\begin{theo}\label{teo}
BagHist, BagFP and BagKDE are $L^2$-consistent.
\end{theo}

\begin{proof}
We will give a global proof, inspired by \cite{hall1997bootstrap} and \cite{scott2015multivariate} to encompass the different methods, because the demonstration for all these estimators follows the same steps. We have to compute:
\begin{enumerate}
\item[(1)] for the expectation $\Expec(\hat{f}(x))=\Expec \left(\Expec [\hat{f} (x)| \mathcal{L}]\right)$ 
\item[(2)] and to calculate the variance we use the decomposition 
$$\text{Var}(\hat{f}(x))=\underbrace{\Expec(\Var(\hat{f}(x)|\mathcal{L}))}_{(A)}+\underbrace{\Var(\Expec(\hat{f}(x)|\mathcal{L}))}_{(B)}$$   

\end{enumerate}
Without loss of generality, in some calculation for histogram or frequency polygon, and with the aim to simplify notations, we will assume that $x \in B_0$.
\begin{enumerate}

\item[(1)] 

\begin{itemize}
\item \emph{BagHist.} If $x \in B_j=\left[ jh, (j+1)h\right]$ , then we have \begin{align*}\Expec [\baghist (x)| \mathcal{L} ] 
  =&  \Expec \left[ \frac{1}{B} \sum_{b=1}^{B} \fstar_b(x) | \mathcal{L} \right]  
  =  \Expec \left[ \fstar_b (x) | \mathcal{L} \right] 
  =  \Expec \left[ 
         \frac{\nu_{j}^{*}}{nh}
       \right] \\
 = &  \frac{\nu_{j}}{nh} =\histo(x)
 \end{align*}
and if $f$ is locally Lipschitz \begin{align*}\left| \Expec (\Expec [\hat{f} (x)| \mathcal{L}]) - f(x)\right|= &\left| \frac{p_j}{h}-f(x)\right|= \left| \frac{h f(\xi_j)}{h}-f(x)\right|\\ \leq & \gamma_j|\xi_j-x| \leq \gamma_j h \to 0 
\end{align*}
\item \emph{BagFP.} If $x \in B_j=\left[ (j-\frac{1}{2})h, (j+\frac{1}{2})h\right]$, then we have

\begin{align*} \Expec [\bagfp (x)| \mathcal{L} ] 
  = & \Expec \left[ \frac{1}{B} \sum_{b=1}^{B} \fstar_b(x) | \mathcal{L} \right]  
  =  \Expec \left[ \fstar_b (x) | \mathcal{L} \right] \\
  = &  \Expec \left[ 
         \left(\frac{1}{2}+j-\frac{x}{h} \right)\frac{\nu_{j}^{*}}{nh}
       + \left(\frac{1}{2} -j+ \frac{x}{h} \right) \frac{\nu_{j+1}^{*}}{nh}|\mathcal{L} \right] \\
  = & \left(\frac{1}{2} +j- \frac{x}{h} \right) \frac{\nu_{j}}{nh}  + \left(\frac{1}{2}-j\ + \frac{x}{h} \right) \frac{\nu_{j+1}}{nh}  =\fp(x)  \end{align*}
and, if $f$ has second derivative, if $x \in B_0$: $$\left|\Expec\left(\Expec [\bagfp (x)| \mathcal{L} ] \right)-f(x)\right| \approx \left|f''(\xi_0)(h^2-3x^2)\right| \leq 4h^2f''(\xi_0)\to 0$$ 

\item \emph{BagKDE.} For BagKDE we have:

\begin{align*}\Expec [\bagkde (x)| \mathcal{L} ] 
  = & \Expec \left[ \frac{1}{B} \sum_{b=1}^{B} \fstar_b(x) | \mathcal{L} \right]  
  =  \Expec \left[ \fstar_b (x) | \mathcal{L} \right] \\
  = &\Expec \left[ \frac{1}{nh}\sum \limits_{i=1}^{n} K\left(\frac{x-x_{i}^{*}}{h} \right) \right]=\kde \end{align*}
  
  and this is well known that $$\left|\Expec\left(\Expec [\bagkde (x)| \mathcal{L} ] \right)-f(x)\right| \to 0 $$
\end{itemize}

\item[(2)] For variance we use the well known formula defined above.

\begin{itemize}
\item \emph{BagHist.} 
\begin{enumerate}
\item[(A)] 
Because of the independence and identical distribution of the bootstrap samples, if $x \in B_0$: \begin{align*}\Var [\baghist (x)| \mathcal{L} ] 
 =  & \Var \left[ \frac{1}{B} \sum_{b=1}^{B} \fstar_b(x) | \mathcal{L} \right]= \frac{1}{B} \Var\left[ \fstar_b(x) | \mathcal{L}  \right]\\ = & \frac{1}{B} \frac{n p_{0}^{*}(1-p_{0}^{*})}{(nh)^2}\end{align*} where $p_{0}^{*}$ is equal to $\frac{\nu_0}{n}$. Taking expectation over $\mathcal{L}$ we have 
 \begin{align*} \frac{1}{Bnh^2}\Expec \left(p_{0}^{*}(1-p_{0}^{*})\right)= &\frac{1}{Bnh^2} \left[\Expec \left(\frac{\nu_0}{n} \right)  - \Var \left(\frac{\nu_0}{n} \right) - \left[\Expec \left(\frac{\nu_0}{n} \right) \right]^2\right] \\  = & \frac{1}{Bnh^2}\left(p_0-\frac{1}{n}p_0-\frac{1}{n}p_{0}^{2}-p_{0}^{2}\right) \\ = & \frac{1}{Bnh^2} \left(hf(\xi_0)-\frac{h}{n}f(\xi_0)-\frac{h^2}{n}f(\xi_0)^2-h^2f(\xi_0)^2\right)\\ = & \frac{1}{Bnh}f(\xi_0)-\frac{1}{Bn^2h}f(\xi_0)-\frac{1}{Bn^2}f(\xi_0)^2- \frac{1}{Bn}f(\xi_0)^2  \end{align*}
 which tends to 0 as $n \to \infty, nh \to \infty$ and $h \to 0$.
\item[(B)] $\Var \left[\Expec(\baghist(x)|\mathcal{L}) \right]= \Var (\histo(x))\leq \frac{p_0}{nh^2}= \frac{f(\xi_0)}{nh} \to 0$
\end{enumerate}

\item \emph{BagFP.}
\begin{enumerate}
\item[(A)] Because of the independence and identical distribution of the bootstrap samples: \begin{align*}
\Var [\bagfp (x)| \mathcal{L} ] 
 & =  \Var \left[ \frac{1}{B} \sum_{b=1}^{B} \fstar_b(x) | \mathcal{L} \right]= \frac{1}{B} \Var\left[ \fstar_b(x) | \mathcal{L}  \right]  \\
 & = \frac{1}{B} \left\{\left(\frac{1}{2}-\frac{x}{h}\right) \Var(\hat{f}^{*}_0) + \left(\frac{1}{2}+\frac{x}{h}\right) \Var(\hat{f}^{*}_1)\right\} \\ & + \frac{2}{B} \left\{ \left(\frac{1}{4}-\frac{x^2}{h^2}\right) \Cov(\hat{f}^{*}_0,\hat{f}^{*}_1) \right\}
\end{align*}
where $\hat{f}^{*}_0$ and $\hat{f}^{*}_1$ are the histogram estimations over $[-h,0]$ and $[0,h]$ respectively. As $\Var(\hat{f}^{*}_0)=\frac{np_{o}^{ *}(1-p_{o}^{ *})}{n^2h^2}$, then taking expectation:
\begin{align*}
\Expec(\Var(\hat{f}^{*}_0)) = &\Expec \left(\frac{n\frac{\nu_0}{n}(1-\frac{\nu_0}{n})}{n^2 h^2}\right)=\frac{1}{n^2h^2} \Expec \left(\nu_0\left(1-\frac{\nu_0}{n}\right)\right)\\ = & \frac{1}{n^2 h^2} \left(\Expec(\nu_0) - \frac{1}{n} \Expec (\nu_0^2)\right)\\ = & \frac{1}{n^2h^2} \left(\Expec(\nu_0) - \frac{1}{n} (\Var(\nu_0)+[\Expec (\nu_0)]^2)\right)\\ = & \frac{1}{n^2h^2} \left[np_0-\frac{1}{n}(np_0(1-p_0)+(np_0)^2) \right]\\ = & \frac{nhf(\xi_0)}{n^2h^2}- \frac{hf(\xi_0)}{n^2h^2}+\frac{h^2f(\xi_0)^2}{n^2h^2}-\frac{nh^2f(\xi_0)^2}{n^2h^2} \to 0
\end{align*}
In the same way $\Expec(\Var(\hat{f}^{*}_1)) \to 0$. As $\Cov(\hat{f}^{*}_0,\hat{f}^{*}_1)=\frac{-np_{0}^{*}p_{1}^{*}}{n^2 h^2}$, then taking expectation we have: 
\begin{align*}
\Expec \left( \Cov(\hat{f}^{*}_0,\hat{f}^{*}_1)\right)= & \frac{1}{nh^2}\Expec\left(\frac{\nu_0}{n}\frac{\nu_1}{n} \right)\leq \frac{1}{n^3h^2}\Expec(\nu_{0}^2)\Expec(\nu_{1}^2)\\ = & \frac{1}{n^3 h^2} [np_0(1-p_0)+p_0^2][np_1(1-p_1)+p_1^2]\\
= & \frac{1}{n^3h^2}(nhf(\xi_0)-nh^2f(\xi_0)^2+h^2f(\xi_0)^2)(nhf(\xi_1)-nh^2f(\xi_1)^2+h^2f(\xi_1)^2)\\ 
= & \frac{1}{n^3h^2}(n^2h^2 f(\xi_0)f(\xi_1)-n^2h^3f(\xi_0)f(\xi_1)^2+nh^3f(\xi_0)f(\xi_1)^2\\ & +\frac{1}{n^3h^2} (n^2h^4f(\xi_0)^2f(\xi_1)^2-n^2h^3f(\xi_0)^2f(\xi_1)-nh^4f(\xi_0)^2f(\xi_1)^2)\\ & +\frac{1}{n^3h^2}(nh^3f(\xi_0)^2f(\xi_1)-nh^4f(\xi_0)^2f(\xi_1)^2+h^4f(\xi_0)^2f(\xi_1)^2)\\ = & \frac{1}{n} f(\xi_0)f(\xi_1)-\frac{h}{n}f(\xi_0)f(\xi_1)^2+\frac{h}{n^2}f(\xi_0)f(\xi_1)^2 \\ & +\frac{h^2}{n}f(\xi_0)^2f(\xi_1)^2 - \frac{h}{n}f(\xi_0)^2f(\xi_1)- \frac{h^2}{n^2}f(\xi_0)^2f(\xi_1)^2\\ & + \frac{h}{n^2}f(\xi_0)^2f(\xi_1)-\frac{h^2}{n^2}f(\xi_0)^2f(\xi_1)^2+\frac{h^2}{n^3}f(\xi_0)^2f(\xi_1)^2 \to 0
\end{align*}

So we conclude that $\Expec \left(\Var [\bagfp (x)| \mathcal{L} ] \right) \to 0$

\item[(B)] We recall from \cite{scott2015multivariate} pag. 103 that  $$\Var(\Expec(\bagfp(x)|\mathcal{L}))=\Var (\fp(x))=\left( \frac{2x^2}{nh^3}+\frac{1}{2nh}\right)f(\xi_0)-\frac{f(\xi_0)^2}{n}+o\left(\frac{1}{n} \right)$$ Then if $nh \to \infty$ and $n \to +\infty$: $$|\Var(\Expec(\bagfp(x)|\mathcal{L}))| \to 0$$
\end{enumerate}

\item \emph{BagKDE.} 

\begin{enumerate}
\item[(A)] Because of the independence and identical distribution of the bootstrap samples: \begin{align*}
\Var [\bagkde (x)| \mathcal{L} ] 
 & =  \Var \left[ \frac{1}{B} \sum_{b=1}^{B} \fstar_b(x) | \mathcal{L} \right]= \frac{1}{B} \Var\left[ \fstar_b(x) | \mathcal{L}  \right] 
 \end{align*}
So we compute $\Var\left[ \fstar_b(x) | \mathcal{L}  \right]$: 
\begin{align*}
\Var\left[ \fstar_b(x) | \mathcal{L}  \right]= & \underbrace{\frac{1}{n} \left[\sum \limits_{i=1}^{n} \frac{1}{(nh)^2} K^2\left(\frac{x-x_{i}^{*}}{h} \right)\right]}_{(a)}-\underbrace{\frac{1}{n^2}\left\{ \frac{1}{nh}\sum \limits_{i=1}^n K\left( \frac{x-x_{i}^{*}}{h}\right)\right\}^2}_{(b)}
\end{align*}
and therefore $\left|\Var\left[ \fstar_b(x) | \mathcal{L}  \right] \right|\leq (a)+(b)$. Taking expectation:
\begin{enumerate}
\item[(a)] \begin{align*}\Expec \left(\frac{1}{n^3h^2} \sum \limits_{i=1}^{n}K^2\left( \frac{x-x_{i}^{*}}{h}\right)\right)= & \frac{1}{n^3h^2}\sum \limits_{i=1}^{n} \Expec \left(K^2\left( \frac{x-x_{i}^{*}}{h}\right) \right)\\ = & \frac{1}{n^2h^2} \sum \limits_{j=1}^{n} K^2\left( \frac{x-x_{ij}}{h}\right) \mathds{P}(x_i^{*}=x_{ij})\\
= & \frac{1}{n^3h^2} \sum \limits_{j=1}^{n} K^2\left( \frac{x-x_{ij}}{h}\right) \leq \frac{1}{n^3h^2} \sum \limits_{j=1}^{n} \tilde{C}\\ = & \frac{\tilde{C}}{(nh)^2} \to 0
\end{align*}
because since $K$ is bounded, $K^2$ also.
\item[(b)]
\begin{align*}
\Expec \left(\frac{1}{n^2} \left (\frac{1}{nh}\sum \limits_{i=1}^{n}K\left( \frac{x-x_{i}^{*}}{h}\right)\right)^2\right)= & \frac{1}{n^4 h^2} \Expec \left( \sum \limits_{i=1}^{n}K\left( \frac{x-x_{i}^{*}}{h}\right)\right)^2\\ = &\frac{1}{n^4h^2} \sum \limits_{j=1}^{n} \left(\sum \limits_{i=1}^{n} K\left( \frac{x-x_{i}^{*}}{h}\right)  \right)^2 \mathds{P}(x_i^{*}=x_{ij})\\ \leq &  \frac{1}{n^5(nh)^2}\sum \limits_{j=1}^{n} \left( \sum \limits_{i=1}^{n} C \right)^2\\ 
= &  \frac{C^2}{n^2(nh)^2} \to 0
\end{align*}

\end{enumerate}

So we conclude that $\Expec \left(\Var [\bagkde (x)| \mathcal{L} ] \right) \to 0$

\item[(B)] It is a well known result that $\Var \left(\Expec [\bagkde (x)| \mathcal{L} ] \right) = \Var (\kde) \to 0$
\end{enumerate}

\end{itemize}

\end{enumerate}
So, with the usual assumption of $n\to \infty, h \to 0, nh \to \infty$ this implies $L^2$ convergence for $\baghist, \bagfp$ and $\bagkde$.
\end{proof}

\section{Experiments}\label{sec:experiments}

We describe in this section a series of numerical experiments aiming to show the 
practical performance of the bagged versions of the classical density estimators. 
First, we obtain a numerical estimate of the MISE on simulated data sets created following baseline densities.
The impact of the aggregation is analyzed. 
We also use the bootstrapped version of the density estimators to 
construct a confidence interval for it and study its performance.

\subsection{Simulations}

Among the numerous possibilities of univariate densities, we choose eight simulation models partially following the work of \cite{Bourel13}. This choice presents a different degree of difficulty related to the number of modes, asymmetry, tail behavior and regularity. We denote them by $\mathcal{M}1$
to $\mathcal{M}8$. Their definition is the object of Table \ref{tab:densities} and Figure \ref{fig:densities} shows a graphical display of the densities.
The notation $\mathcal{N}(\mu, \sigma^2)$ is used to refer to a normal distribution with mean equal to $\mu$
and variance equal to $\sigma^2$, $\mathcal{U}[a, b]$ is the uniform density over the support $[a, b]$,
and $\chi^2_\nu$ is a Chi squared density with $\nu$ degrees of freedom.
Models 3, 4, 7 and 8 are mixtures of densities. Models 2 and 7 are asymmetrical.

\begin{table}[ht] \centering
\begin{tabular}{ll} \hline
Model            & Description \\ \hline
($\mathcal{M}1$) : Normal Standard & Standard Gaussian density $\mathcal{N}(0,1)$ \\
($\mathcal{M}2$) : Chi 10 & Chi-square density $\chi^2_{10}$\\
($\mathcal{M}3$) : Mix1 & $0.5 \mathcal{N}(-1,0.3)+0.5\mathcal{N}(1,0.3)$\\
($\mathcal{M}4$) : Claw & the Claw Density \citep{Marron}\\
($\mathcal{M}5$) : Triangular & Symmetric triangular density with support on [0,2]\\ 
($\mathcal{M}6$) : Uniform 0-1  & Uniform density $\mathcal{U}[0, 1]$\\
($\mathcal{M}7$) : Mix2 & $0.5 \mathcal{N}(0,1) + 0.5 \sum \limits_{i=1}^{10} \mathbf{1}_{\left(\frac{2(i-1)}{10},\frac{2i-1}{10}  \right]}$ \citep{Rigo} \\ 
($\mathcal{M}8$) : Mix3 & Mixture of uniforms $0.5 \mathcal{U}[-2, -1] + 
                          0.5 \mathcal{U}[1, 2]$ \\ \hline
\end{tabular}
\caption{Simulated univariate densities.} \label{tab:densities}
\end{table}

At each replication we draw two datasets following each density. The first one is used for estimation purposes while the second one is left-out for evaluation (either MISE or empirical covering). All the simulations are done with the \texttt{R} software (\cite{elR}).

\begin{figure}[ht]  \centering
\includegraphics[width=.9\textwidth]{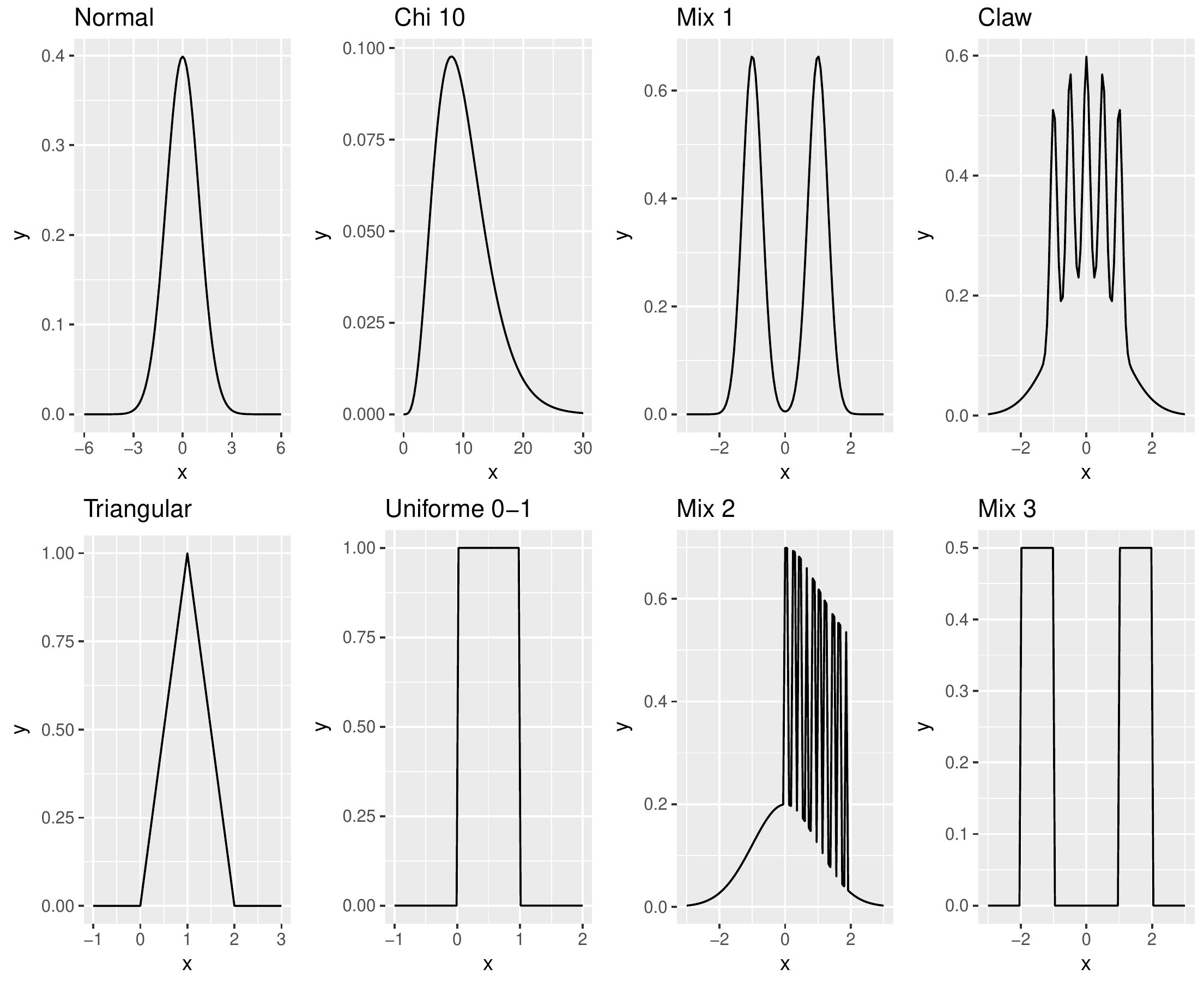}
\caption{Densities used for the simulations.} \label{fig:densities}
\end{figure}

\subsection{Quality of the estimation}

We compare density estimators of different nature. On one hand side we use three individual estimators: histograms (H), 
frequency polygons (FP)  and kernel density estimators (KDE), on the other hand, their bagged versions, respectively BagHist, BagFP and BagKDE.
Also we include the RASH estimator. We use cross validation to calibrate the bandwidth $h$ at each step of all the intermediate estimation methods. An alternative would be to use 
maximum likelihood as in \cite{Bourel13}. In our framework cross validation has, in general, a better computation behavior. Also it is more
general and may be used for example with dependent data as in time series.

\begin{figure}[h]
\includegraphics[width = \textwidth]{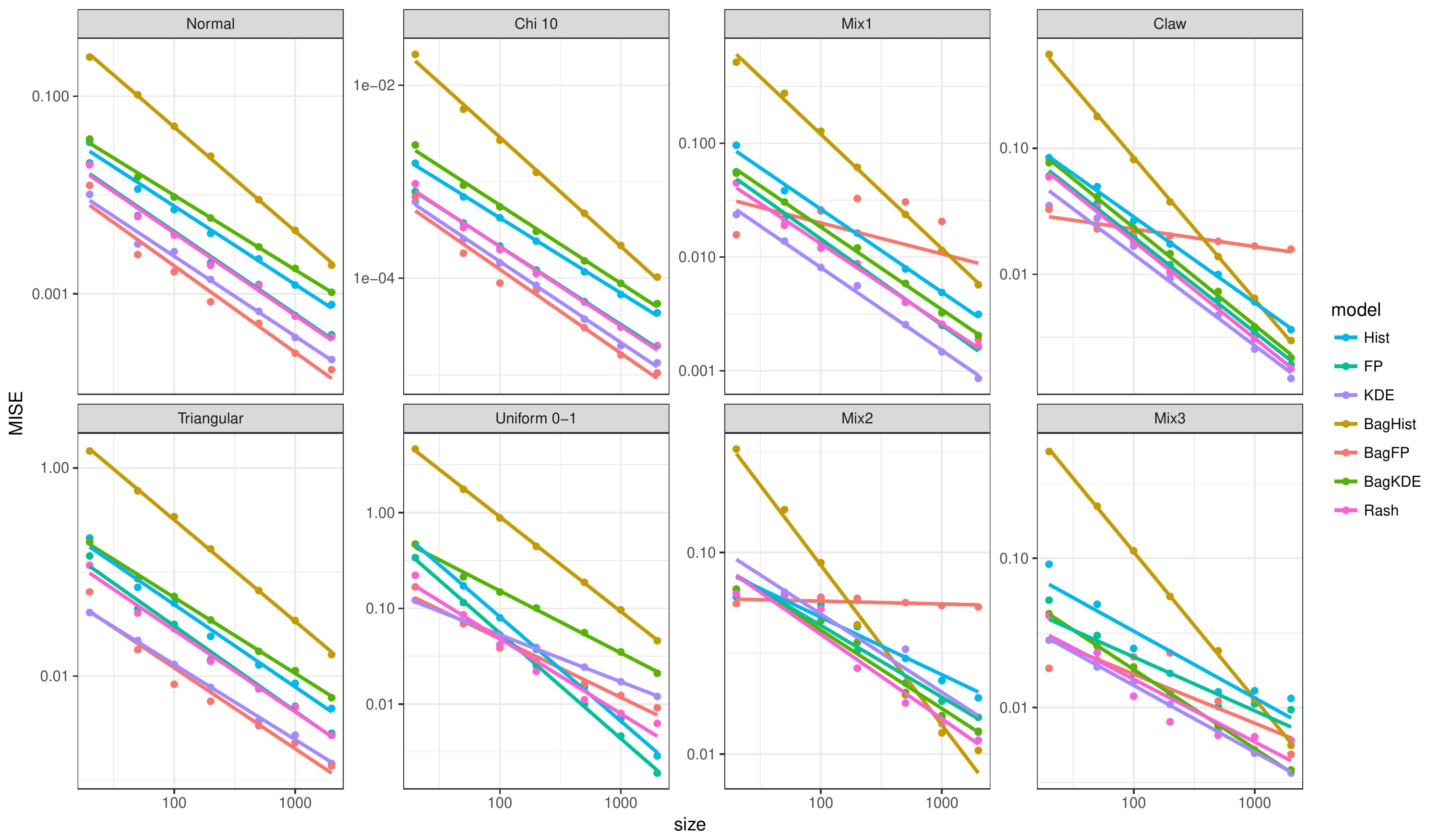}
\caption{\label{fig:mise} MISE by estimation method for the six simulation data sets in scale log.
}
\end{figure}

Figure \ref{fig:mise} represents the dependence of MISE on the sample size $n$ for the different combinations of
densities and estimators. Each point represents the average of $M=100$ times the MISE$\times 100$ of the method using $B = 200$ 
intermediate estimators for the aggregating methods.
Notice that these plots are in log-log scale which is useful to highlight the convergence behaviour.
Individual values of these plots are presented in Appendix \ref{sec:addquality}.

Let us comment these plots. First, the adjusted lines are of relative good quality since the points corresponding to each combination
density-estimator are almost aligned. Remember that each point is the mean average of $M = 100$ replicates and so the 
inner replicate variability is reduced even for a few points determining each line. 
However, in some few cases the quality of the fit is quite poor. Now, for each panel most of the adjusted
lines are almost parallels which means the methods share a similar convergence behavior. 
Comparing the four simple densities (leftmost panels) and the four mixtures (rightmost panels)
a difference in the behaviour seems to appear, at least if one looks at some bagged version
as for instance the BagFP.
Second, if one compares each individual estimator
(H, FP and KDE) with their bagged versions, the latter success to reduce the MISE in most of the situations. An important exception if BagHist which produces almost always worst results than the intermediate estimator, i.e. the histogram. However, the asymptotic behaviour
is such that with large sample sizes it is able to catch the quality of the histogram and
even overwhelm it on the four mixtures.
On regular targets, KDE (or at least its bagged version) shows a very competitive performance. However, in presence of multi-modality they loss relative competitiveness with large sample sizes. The fact that the results are not entirely satisfactory for the bagged version of KDE may be because KDE it is already a good and stable density estimator (more stable in any case than histogram) and, according with \cite{Brei-bag}, bagging kernel density estimators may be degrade the performance of this stable procedure.

\subsection{Reduction of MISE due to aggregation}

We concentrate now on aggregating methods. A natural matter to look at is the quality
of the aggregation as the number of bootstrap samples increases. For this, we examine the MISE of the bagged versions
for a range of increasing bootstrap samples.
We replicate $M =100$ times each combination of density simulation to construct the different curves.
The result of experiments are presented in Figure \ref{fig:miseagg} in a log-log scale with $n=500$ observations. 

Globally we observe that MISE decreases monotonically with increasing values of $B$ until some point between 20 and 50 bootstrap samples
after which more samples does not produce further enhancement. 

\begin{figure}[h] \centering
\includegraphics[width=\textwidth]{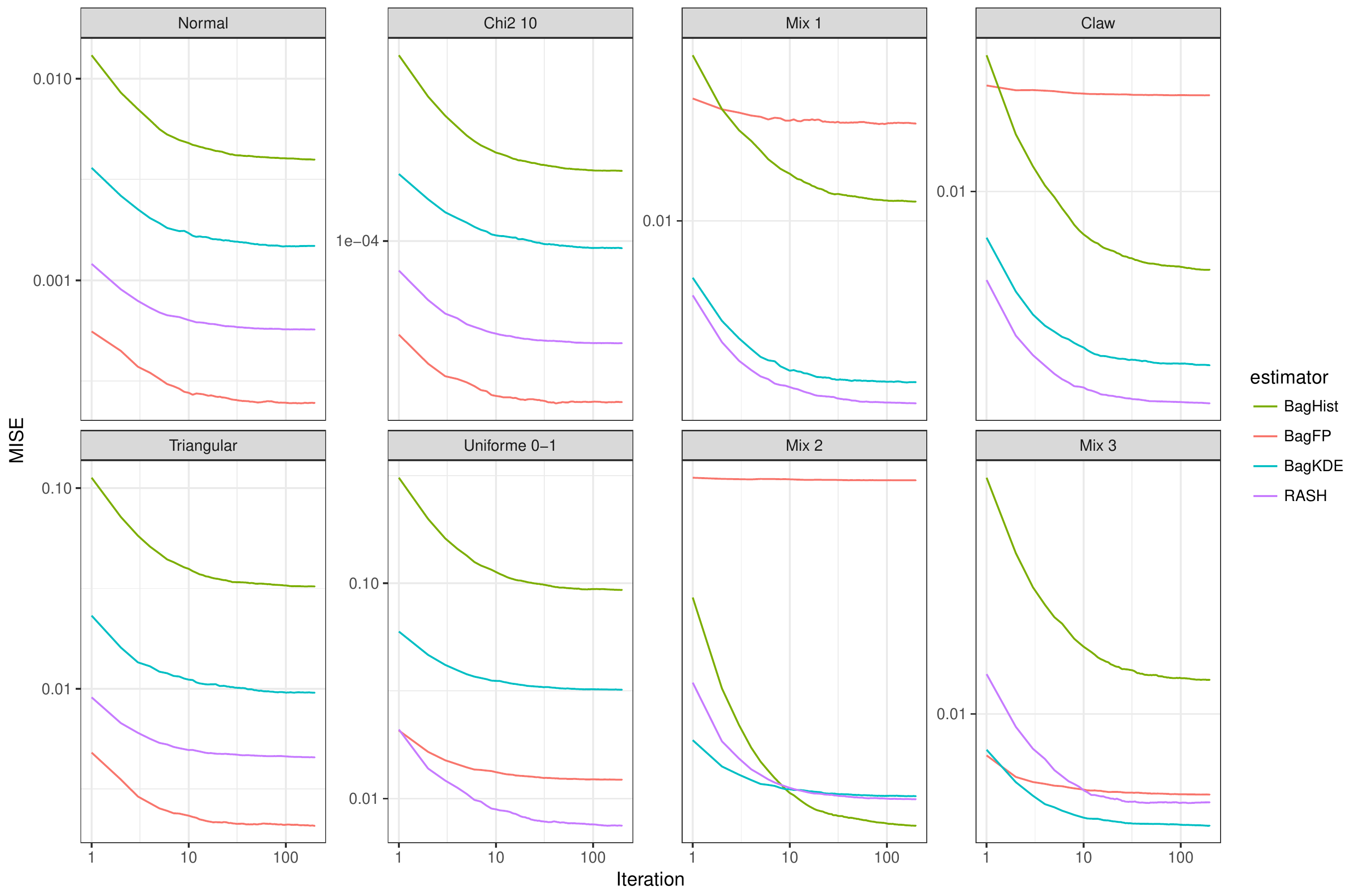}
\caption{\label{fig:miseagg} MISE error vs number of aggregates, n=500, M=100, B=200 in log-log scale.}
\end{figure}

\subsection{Variability bands}

A natural by-product of bootstrap samples is the construction of confidence bands.
For some level $\alpha$, one wants to estimate the quantities $\hat{l}_n(x)$ and $\hat{u}_n(x)$ that 
verify 
\[\Prob \lbrace{ \hat{l}_n(x) \leq f (x) \leq \hat{u}_n(x) \rbrace} \geq 1 - \alpha,\,\,\,\,\,\, \forall\,x\]
that is, the quantities are the borders of an interval that covers at the true density $f(x)$ at some confidence
level $(1 - \alpha) \times 100 \%$.
We tackle here its construction for the density estimator. Generally, a confidence band for $f$ is centered over an estimator $\hat{f}_n$ of $f$ and has the form $\hat{f}_n(x) \pm c \hat{\sigma_n}(x)$ for all $x$, with $c>0$. However, 
since nonparametric density estimators are biased, 
the usual construction does not yields on a really a confidence band for $f$.
Indeed, for a fixed $x$, due to the bias $\Expec \left(\hat{f}_n(x)\right) - f(x)$, it is not easy to derive a confidence interval 
using the pivotal quantity $\frac{\hat{f}_n(x)-f(x)}{\hat \sigma(x)}$. So, the interval is usually centered at $\overline{f}_n(x) = \Expec \left(\hat{f}_n(x)\right)$ instead of being around $f(x)$. For this reason, these confidence bands are often called \emph{variability bands}. 

We describe two popular constructions to compare with our procedure.

\begin{enumerate}
\item 
\emph{Variability Band for histograms.} Under mild conditions \citep{wasserman2006all}
the histogram estimator $\histo (x)$ is approximately unbiased for the target density $f(x)$. But the approximate variance is $f(x)/(nh)$ where $h = 1 / m$ is the inverse of the number of bins $m$. Its dependence on the unknown target is an obstacle. 
To circumvent it, \cite{scott2015multivariate} looks at $\Var \left[\sqrt{\histo (x) }\right]$ which is approximately $1 / (4 n h)$ and thus independent of $f(x)$. We define $\histom = \Expec [\histo (x)]$ as the target and as we say before the confidence band will not take account of the bias but only of the variability of the estimator. Then, using a normal approximation it is easy to
show that \citep[p. 130]{wasserman2006all}:
\[l_n(x)= \left(\max\left\{ \sqrt{\histo(x)} -c, 0\right\}\right)^2, \qquad
u_n(x)=\left(\sqrt{\histo(x)}+c \right)^2
\]
where $c= \frac{z_{\alpha/(2m)}}{2}\sqrt{\frac{m}{n}}$ give an approximate variability band for $\histom$ at $(1 - \alpha) \times 100\%$ of confidence.
\item 
\emph{Variability Band for KDE.} As we have shown with the histogram, since the variance of $\kde (x)$ also involves the true density $f$, it is more suitable to use the square root 
(see \cite{bowman1997applied}). In the case of the kernel density estimator $\text{Var}\left(\sqrt{\kde(x)}\right) \approx \frac{||K||_2^2}{4nh_n}$ and again does not depend on the 
true unknown density $f$. On this square root scale, for a fixed point $x$ we consider the interval that back to the original scale is given by

\[l_n(x) = \left(  \kde(x) - \frac{||K||_2}{\sqrt{4nh_n}}  \right)^2, \qquad
  u_n(x) = \left(  \kde(x) + \frac{||K||_2}{\sqrt{4nh_n}}  \right)^2.
\]
As we said before this is not a confidence band for the true density $f$, because of the bias so we will talk about a variability band.

\item \emph{Bootstrap based confidence interval and resulting tube}
The bootstrapped sample induces a distribution that can be used to asses the variability of the estimator. Indeed, the simple superposition of the individual estimators (histogram, frequency polygon, kernel density estimator) gives a coarse idea of the uncertainty around the aggregate estimator. More the scatter of individual individual density estimators is dispersed, higher is the variance of the estimator. For the concrete construction of the confidence band we first fix the abscissa $x \in \mathbb{R}$. Then we consider
the set of bootstrapped density estimators evaluated at that point, i.e. $\lbrace \fstar_1(x), \ldots, \fstar_B(x) \rbrace$. 
This set is a collection of $B$ univariate measures. Note that the bagged estimation is the mean average of this collection. Then, a $(1 - \alpha) \times 100 \%$ confidence interval can be obtained by considering the empirical quantiles at $\alpha/2$ and $1 - \alpha/2$ of this ensemble. 
\end{enumerate}

\begin{figure}[!ht]
\begin{center}
\includegraphics[width=\textwidth]{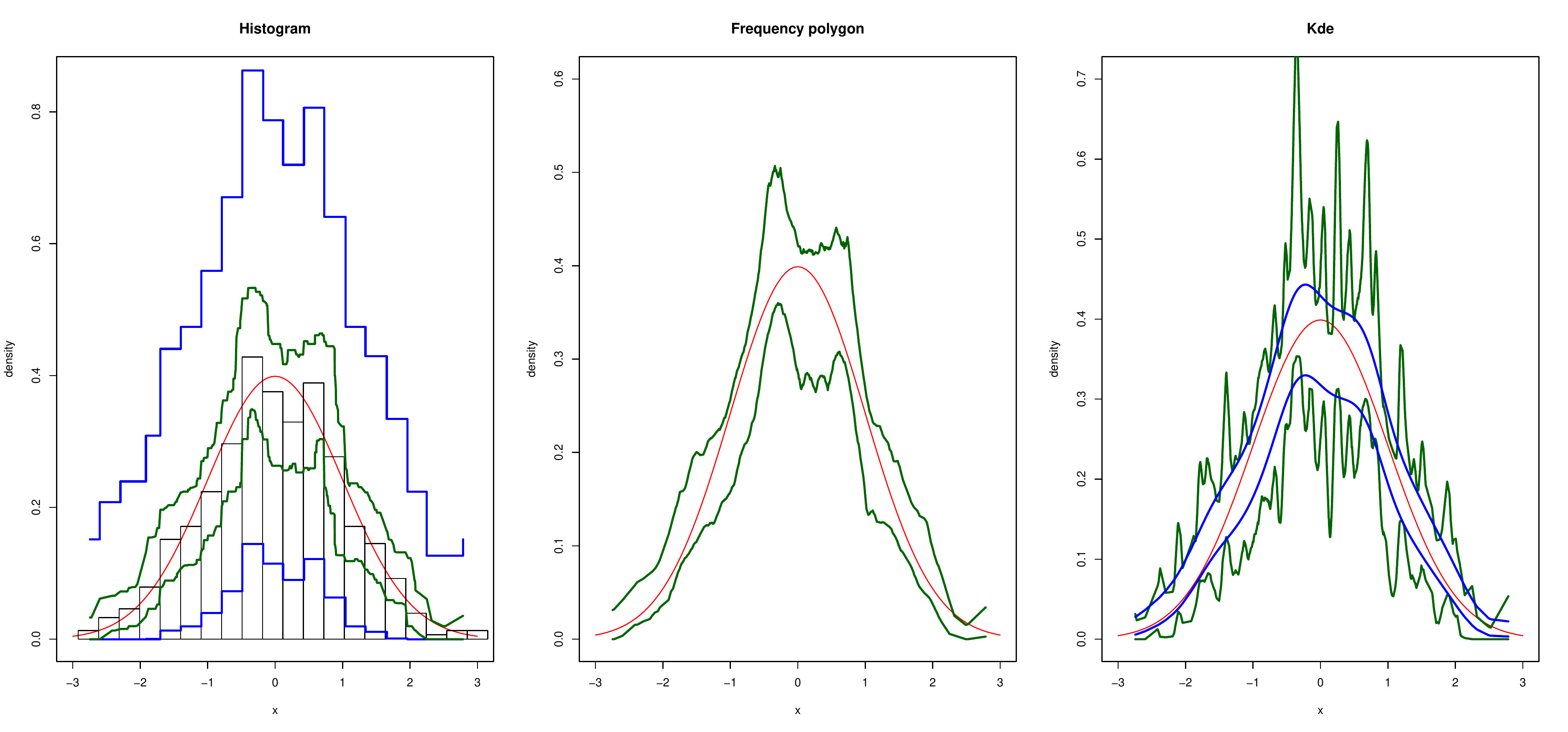}
\caption{Variability bands for the estimation of a standard normal target (red)
based on Histograms (left), Frequency polygons (middle) and KDE (right).
Confidence intervals are either constructed for individual estimators (in blue) when available or constructed using bootstrap sampling (in green).
}\label{BConfHFPKDE}
\end{center}
\end{figure}

In Figure \ref{BConfHFPKDE} we show the different constructions generated by these methods for a standard normal target. Each panel correspond to one of the three intermediate density estimators (Histogram, frequency polygons and KDE). 
At each time, the bootstraped construction for the confidence interval is 
represented. For the histogram and KDE we also draw the constructions detailed above. The confidence interval for the histogram succeeds to cover the true density but produces an arguably too large band. For KDE the confidence for the intermediate
estimator is very good, and the one obtained by bootstrap sampling presents a relatively high variability.  

We compare the alternative constructions of the confidence band using two metrics. 
The aim is to obtain the narrowest band that warranties a given nominal coverage.
For this, we consider the empirical coverage of the bands and its mean width.
Let us call $\hat l_n (t_i)$ and $\hat u_n(t_i)$ the lower and upper bounds of the 
confidence bands, evaluated at points $t_i, i = 1, 2, \ldots, N$. Then, we call
the empirical mean coverage of the target $f(x)$ the quantity
\[ \frac{1}{N} \sum_{i = 1}^{N}
              \mathbb{I}_{\lbrace \hat l_n (t_i) \leq f(t_i) \leq \hat u_n(t_i)\rbrace},  \ \]
where $\mathbb{I}_A$ is the indicator function of the set $A$. The mean width of the interval
is defined by 

\[ \frac{1}{N} \sum_{i = 1}^{N} \left\lbrace
             \hat u_n(t_i) - \hat l_n (t_i) 
             \right\rbrace.
\]

In this experiment we set the confidence level at $95\%$ to construct the 
variability bands. 
We give as reference the variability band constructed through the kernel density estimator as explained before (we denote this method as KDE-sm).
The construction for the individual histograms produces too larges bands which
always cover the true density. For this reason they are not presented in Table \ref{tab:tubos}. 

In general, the bands cover reasonably well the simplest densities (on the top rows) having more difficulties with the more exotic density models (on the bottom rows). Particularly, the last three densities in the table are too difficult targets producing very low empirical coverings. If we look at the methods, the construction using bootstraped histograms is the best one among the boostrap based ones, and give a fair competitor to KDE-sm.

\begin{table}[h] 
\centering
\begin{tabular}{|l|cccc|cccc|} \hline
            & \multicolumn{4}{c}{Coverage} & \multicolumn{4}{|c|}{Mean width}	\\
Density  	& Hist  &   FP	& KDE   & KDE-sm & Hist & FP &	KDE & KDE-sm \\  \hline
Normal      & 95.45 & 92.54 & 92.83 & 96.54 & 0.24 & 0.18 & 0.19 & 0.09 \\ 
Chi10        & 95.75 & 92.88 & 92.96 & 94.79 & 0.06 & 0.04 & 0.04 & 0.02 \\ 
Mix1        & 95.55 & 93.11 & 90.34 & 95.27 & 0.40 & 0.30 & 0.25 & 0.19 \\ 
Claw        & 95.05 & 89.79 & 87.07 & 91.33 & 0.30 & 0.22 & 0.25 & 0.22 \\ 
Triangular  & 95.79 & 92.74 & 92.69 & 94.94 & 0.66 & 0.51 & 0.45 & 0.23 \\ 
Uniform 0-1    & 92.20 & 89.40 & 88.61 & 90.56 & 1.06 & 0.84 & 0.73 & 0.52 \\ 
Mix2       & 77.78 & 41.00 & 63.14 & 47.73 & 0.36 & 0.26 & 0.26 & 0.21 \\ 
Mix3       & 88.86 & 85.59 & 83.41 & 89.84 & 0.37 & 0.29 & 0.22 & 0.27 \\ 
 \hline
\end{tabular}
\caption{\label{tab:tubos} Mean empirical coverage and mean interval widths for the densities and estimators
considered.} 
\end{table}

\section{Conclusions\label{sec:discussion}}
In this work we present three univariate density estimators obtained by aggregation such as in Bagging. For each method, the intermediate estimators are histograms, frequency polygons or kernel density estimators. We prove the $L^2$ consistency of the three estimators and do several simulations over densities with different characteristics. 
Also, we bring a way to compute a kind of confidence band, which is more close to a point wise variability band in the sense that the authors who studied on this subject give. This construction needs a deeper study to be able to draw more conclusive conclusions about it. Another clue for future work is to investigate the natural extension of considering the bagged construction over multivariate  densities. 


\subsection*{Acknowledgements}
We would like to thank project ECOS-2014 \emph{Aprendizaje Autom\'atico para la Modelizaci\'on y el An\'alisis de Recursos Naturales}, n$^{o}$ U14E02, the \emph{LIA-IFUM} and the \emph{ANII} -Uruguay for their financial support.

\bibliographystyle{plain}      
\bibliography{references}

\appendix

\section{Additional results} 

\subsection*{Quality of the estimation\label{sec:addquality}}

For sake of completeness we present in this appendix the individual values of Figure \ref{fig:mise}.
In the following tables (Tables 3 to 7), values are $100\times$MISE obtained as mean average over 100 replicates.
At each line, best results are  shown in blue.

\def\best{\color{blue}}
\def\worst{\color{red}}

\begin{table}[ht]
\centering
\begin{tabular}{rrrrrrrr}
  \hline
 & Hist & FP & KDE & BagHist & BagFP & BagKDE & RASH \\ 
  \hline
Normal & 1.1447 & 0.6171 & 0.3181 & 10.3243 & \best 0.2487 & 1.5195 & 0.6025 \\ 
  Chi10 & 0.0696 & 0.0374 & 0.0254 & 0.5635 & \best 0.0181 & 0.0919 & 0.0334 \\ 
  Mix1 & 3.8260 & 2.1742 & \best 1.3801 & 27.4794 & 1.8937 & 3.0423 & 1.9464 \\ 
  Claw & 4.9526 & 3.6531 & 2.7909 & 17.8066 & \best 2.2902 & 4.1301 & 3.4380 \\ 
  Triangular & 7.1220 & 4.4168 & 2.2008 & 60.1150 & \best 1.7936 & 8.6147 & 4.0222 \\ 
  Uniform 0-1 & 17.2112 & 11.4841 & \best 6.8795 & 175.0100 & 6.9725 & 21.2978 & 8.5623 \\ 
  Mix2 & 6.2811 & 6.0856 & 6.3647 & 16.3470 & 5.8546 & \best 5.7913 & 5.9907 \\ 
  Mix3 & 4.9186 & 3.0408 & \best 1.8700 & 22.3239 & 2.0799 & 2.5847 & 2.3357 \\ 
   \hline
\end{tabular}
\caption{MISE for sample size $n=50$}
\end{table}

\begin{table}[ht]
\centering
\begin{tabular}{rrrrrrrr}
  \hline
 & Hist & FP & KDE & BagHist & BagFP & BagKDE & RASH \\ 
  \hline
Normal & 0.7098 & 0.4016 & 0.2665 & 4.9699 & \best 0.1665 & 0.9510 & 0.3908 \\ 
  Chi10 & 0.0422 & 0.0215 & 0.0144 & 0.2689 & \best 0.0089 & 0.0549 & 0.0198 \\ 
  Mix1 & 2.5451 & 1.3417 & \best 0.8085 & 12.7401 & 2.5901 & 1.9115 & 1.1964 \\ 
  Claw & 2.6535 & 1.9604 & \best 1.6735 & 8.1400 & 2.1407 & 2.3149 & 1.8267 \\ 
  Triangular & 5.0242 & 3.1358 & 1.2952 & 33.9450 & \best 0.8356 & 5.8185 & 2.8512 \\ 
  Uniform 0-1 & 7.9732 & 5.2470 & 5.2962 & 87.8440 & \best 3.8460 & 14.8139 & 4.1718 \\ 
  Mix2 & 5.8523 & 5.5056 & 6.0074 & 8.9022 & 5.9432 & \best 4.5841 & 5.2529 \\ 
  Mix3 & 2.4888 & 1.6908 & 1.4754 & 11.2250 & 2.1891 & 1.8702 & \best 1.1925 \\ 
   \hline
\end{tabular}
\caption{MISE for sample size $n=100$}
\end{table}

\begin{table}[ht]
\centering
\begin{tabular}{rrrrrrrr}
  \hline
 & Hist & FP & KDE & BagHist & BagFP & BagKDE & RASH \\ 
  \hline
Normal & 0.4077 & 0.2059 & 0.1394 & 2.4724 & \best 0.0827 & 0.5840 & 0.1937 \\ 
  Chi10 & 0.0242 & 0.0121 & 0.0084 & 0.1246 & \best 0.0075 & 0.0305 & 0.0111 \\ 
  Mix1 & 1.6231 & 0.8703 & \best 0.5597 & 6.1511 & 3.2618 & 1.2028 & 0.8753 \\ 
  Claw & 1.7414 & 1.1899 & \best 0.9338 & 3.7563 & 2.0185 & 1.4601 & 1.0360 \\ 
  Triangular & 2.4070 & 1.4362 & 0.7807 & 16.6818 & \best 0.5730 & 3.4612 & 1.3875 \\ 
  Uniform 0-1 & 3.7493 & 2.5354 & 3.9036 & 44.4120 & 2.7891 & 10.1004 & \best 2.1922 \\ 
  Mix2 & 4.2920 & 3.2531 & 5.7958 & 4.3796 & 5.9281 & 3.5878 & \best 2.6659 \\ 
  Mix3 & 1.6959 & 1.2517 & 1.0518 & 5.5502 & 2.3213 & 1.2049 & \best 0.8032 \\ 
   \hline
\end{tabular}
\caption{MISE for sample size $n=200$}
\end{table}

\begin{table}[ht]
\centering
\begin{tabular}{rrrrrrrr}
  \hline
 & Hist & FP & KDE & BagHist & BagFP & BagKDE & RASH \\ 
  \hline
Normal & 0.2254 & 0.1236 & 0.0663 & 0.8958 & \best 0.0500 & 0.2977 & 0.1203 \\ 
  Chi10 & 0.0116 & 0.0057 & 0.0038 & 0.0470 & \best 0.0031 & 0.0151 & 0.0056 \\ 
  Mix1 & 0.7855 & 0.3997 & \best 0.2542 & 2.3672 & 3.0443 & 0.5857 & 0.4018 \\ 
  Claw & 0.9964 & 0.6338 & \best 0.4747 & 1.3815 & 1.8236 & 0.7305 & 0.5522 \\ 
  Triangular & 1.2801 & 0.7536 & 0.3709 & 6.6334 & \best 0.3338 & 1.7276 & 0.7614 \\ 
  Uniform 0-1 & 1.5125 & \best 0.9942 & 2.4329 & 18.6821 & 1.6281 & 5.5846 & 1.1093 \\ 
  
  Mix2 & 2.9831 & 2.2515 & 3.3141 & 1.9641 & 5.6468 & 2.0181 & \best 1.7904 \\ 
  Mix3 & 1.2756 & 1.0160 & 0.7029 & 2.3993 & 1.1044 & 0.7520 & \best 0.6518 \\ 
   \hline
\end{tabular}
\caption{MISE for sample size $n=500$}
\end{table}

\begin{table}[ht]
\centering
\begin{tabular}{rrrrrrrr}
  \hline
 & Hist & FP & KDE & BagHist & BagFP & BagKDE & RASH \\ 
  \hline
Normal & 0.1220 & 0.0599 & 0.0360 & 0.4380 & \best 0.0249 & 0.1804 & 0.0591 \\ 
  Chi10 & 0.0068 & 0.0031 & 0.0020 & 0.0218 & \best 0.0016 & 0.0088 & 0.0031 \\ 
  Mix1 & 0.4877 & 0.2512 & \best 0.1465 & 1.1433 & 2.0536 & 0.3230 & 0.2569 \\ 
  Claw & 0.6064 & 0.3468 & \best 0.2559 & 0.6452 & 1.6762 & 0.3832 & 0.3082 \\ 
  Triangular & 0.8526 & 0.5124 & 0.2713 & 3.4204 & \best 0.2305 & 1.1284 & 0.5016 \\ 
  Uniform 0-1 & 0.7052 & \best 0.4630 & 1.7129 & 9.6299 & 1.2285 & 3.5032 & 0.7965 \\ 
  Mix2 & 2.3140 & 1.8291 & 1.5237 & \best 1.2728 & 5.4538 & 1.5519 & 1.4157 \\ 
  Mix3 & 1.2970 & 1.0580 & \best 0.4948 & 1.1314 & 0.6291 & 0.5210 & 0.6376 \\ 
   \hline
\end{tabular}
\caption{MISE for sample size $n=1000$}
\end{table}

\end{document}